\documentclass[a4paper,USenglish,numberwithinsect]{lipics}
\usepackage{graphicx}
\usepackage{thmtools}
\usepackage{enumerate}
\usepackage{tikz}
\usetikzlibrary{arrows,positioning,automata}
\usepackage{multirow}
\usepackage{hyperref}

\theoremstyle{plain}
\newtheorem{lemma-indexed-by-theorem}[theorem]{Lemma}
\theoremstyle{definition}
\newtheorem{example-indexed-by-theorem}[theorem]{Example}

\newcommand{\ceil}[1]{\left\lceil #1\right\rceil}
\newcommand{\suchthat}{\:|\:}
\newcommand{\N}{\mathbb{N}}
\newcommand{\Q}{\mathbb{Q}}
\renewcommand{\P}{\mathsf{P}}
\newcommand{\sharpP}{\ensuremath{\mathsf{\# P}}}
\newcommand{\NP}{\ensuremath{\mathsf{NP}}}

\newcommand{\ApproxMC}{\ensuremath{\mathsf{ApproxMC}}}
\newcommand{\UniGen}{\ensuremath{\mathsf{UniGen}}}

\newcommand{\improvs}{I}
\newcommand{\valids}{A}
\newcommand{\creative}{admissible}
\newcommand{\wref}{w_{\text{ref}}}

\begin{document}

\title{Control Improvisation\footnote{This is a preliminary version of the paper, which has since been substantially extended: please see Fremont et al.~2017 \cite{jacm-preprint}.}}
\author{Daniel J. Fremont}
\author{Alexandre Donz\'e}
\author{Sanjit A. Seshia}
\author{David Wessel}
\affil{University of California, Berkeley\\\{\texttt{dfremont}, \texttt{donze}, \texttt{sseshia}\}\texttt{@berkeley.edu}}
\Copyright{Daniel J. Fremont, Alexandre Donz\'e, Sanjit A. Seshia, and David Wessel}

\maketitle

\begin{abstract} 
We formalize and analyze a new automata-theoretic
problem termed {\em control improvisation}. Given an automaton, the
problem is to produce an \emph{improviser}, a probabilistic algorithm
that randomly generates words in its language, subject to two
additional constraints: the satisfaction of an \emph{admissibility}
predicate, and the exhibition of a specified amount of
randomness. Control improvisation
has multiple applications, including, for example, 
generating musical improvisations
that satisfy rhythmic and melodic constraints, where admissibility is
determined by some bounded divergence from a reference melody. We
analyze the complexity of the control improvisation problem, giving
cases where it is efficiently solvable and cases where it is
\sharpP-hard or undecidable. We also show how symbolic techniques
based on Boolean satisfiability (SAT) solvers can be used to approximately solve some of the
intractable cases.
\end{abstract} 

\section{Introduction}
\label{sec:intro}

We introduce and formally characterize a new automata-theoretic problem termed {\em control
improvisation}. Given an automaton, the problem is to produce an \emph{improviser}, a probabilistic
algorithm that randomly generates words in the language of the automaton, subject to two additional
constraints: each generated word must satisfy an \emph{admissibility} predicate, and the improviser
must exhibit a specified amount of randomness.

The original motivation for this problem arose from a topic known as {\em machine improvisation of
music}~\cite{rowe-2001}.  Here, the goal is to create algorithms which can generate variations of a
reference melody like those commonly improvised by human performers, for example in jazz.  Such an
algorithm should have three key properties.  First, the melodies it generates should conform to rhythmic
and melodic constraints typifying the music style (e.g. in jazz, the melodies should follow the
harmonic conventions of that genre). Second, the algorithm should be sufficiently randomized that
running it several times produces a variety of different improvisations. Finally, the generated
melodies should be actual variations on the reference melody, neither reproducing it exactly nor being
so different as to be unrecognizable. In previous work~\cite{donze-icmc14}, we identified these
properties in an initial definition of the control improvisation problem, and applied it to the
generation of monophonic (solo) melodies over a given jazz song harmonization\footnote{Examples of
improvised melodies can be found at the following URL:\\
\url{http://www.eecs.berkeley.edu/~donze/impro_page.html}.}.
 
These three properties of a generation algorithm are not specific to music.  Consider
\emph{black-box fuzz testing} \cite{fuzzing-book}, which produces many inputs to a program hoping to
trigger a bug.  Often, constraints are imposed on the generated inputs, e.g. in \emph{generative}
fuzz testing approaches which enforce an appropriate format so that the input is not rejected
immediately by a parser.  Also common are \emph{mutational} approaches which guide the generation
process with a set of real-world seed inputs, generating only inputs which are variations of those
in the set. And of course, fuzzers use randomness to ensure that a variety of inputs are tried.
Thus we see that the inputs generated in fuzz testing have the same general requirements as music
improvisations: satisfying a set of constraints, being appropriately similar/dissimilar to a
reference, and being sufficiently diverse.

We propose control improvisation as a precisely-defined theoretical problem capturing these
requirements, which are common not just to the two examples above but to many other generation
problems.  Potential applications also include home automation mimicking typical occupant behavior (e.g.,
randomized lighting control obeying time-of-day constraints and limits on energy
usage~\cite{lee-personal13}) and randomized variants of the supervisory control problem
\cite{lafortune06}, where a controller keeps the behavior of a system within a safe operating region
(the language of an automaton) while adding diversity to its behavior via randomness.  A typical example
of the latter is surveillance: the path of a patrolling robot should satisfy various constraints
(e.g. not running into obstacles) and be similar to a predefined route, but incorporate some
randomness so that its location is not too predictable \cite{lafortune-personal15}.

Our focus, in this paper, is on the {\em theoretical characterization of
control improvisation}. Specifically, we give a precise theoretical
definition and a rigorous characterization of the complexity of the
control improvisation problem under various conditions on the inputs
to the problem. While the problem is distinct from any other we have
encountered in the literature, our methods are closely connected to
prior work on random sampling from the languages of automata and
grammars \cite{hickey-cohen,denise2006,sharpNFA}, and sampling from
the satisfying assignments of a Boolean formula~\cite{unigen}.
Probabilistic programming techniques~\cite{probprog} could be used
for sampling under constraints, but the present methods cannot be used
to construct improvisers meeting our definition.

In summary, this paper makes the following novel contributions:
\begin{itemize}
\item Formal definitions of the notions of control improvisation (CI) and a
polynomial-time improvisation scheme (Sec.~\ref{sec:prelim});
\item A theoretical characterization of the conditions under which
improvisers exist (Sec.~\ref{section:existence});
\item A polynomial-time improvisation scheme for a practical class of
CI instances, involving finite-memory admissibility predicates
(Sec.~\ref{section:finite-memory}); 
\item \sharpP-hardness and undecidability results for more general classes of
the problem (Sec.~\ref{section:complex-automata}); 
\item A symbolic approach based on Boolean satisfiability (SAT)
solving that is useful in the case when the automata
are finite-state but too large to represent explicitly (Sec.~\ref{section:symbolic}).
\end{itemize}
We conclude in Sec.~\ref{section:conclusion} with a
synopsis of results and directions for future work.
For lack of space, we include only selected proofs and proof sketches in the main body
of the paper; complete details may be found in the Appendix.

\section{Background and Problem Definition}
\label{sec:prelim}

In this section, we first provide some background on a previous
automata-theoretic method for music improvisation based on a data
structure called the {\em factor oracle}. We then provide a formal
definition of the control improvisation problem while explaining the
choices made in this definition.

\subsection{Factor Oracles}

An effective and practical approach to machine improvisation of music (used for example in the prominent OMax system \cite{omax}) is based on a data structure
called the factor oracle \cite{AssayagD04,Cleophas03constructingfactor}. Given a word
$\wref$ of length $N$ that is a symbolic encoding of a reference melody, a factor oracle $F$ is an
automaton constructed from $\wref$ with the following key properties: $F$ has $N+1$ states, all
accepting, chained linearly with direct transitions labelled with the letters in $\wref$, and with
potentially additional forward and backward transitions. Figure \ref{figure:factor-oracle} depicts
$F$ for $\wref = bbac$.  A word $w$ accepted by $F$ consists of concatenated ``factors'' of $\wref$,
and its dissimilarity with $\wref$ is correlated with the number of non-direct transitions. By
assigning a small probability $\alpha$ to non-direct transitions, $F$ becomes a generative Markov
model with tunable ``divergence'' from $\wref$. In order to impose more musical structure on the
generated words, our previous work~\cite{donze-icmc14} additionally requires that improvisations
satisfy rules encoded as deterministic finite automata, by taking the product of the generative
Markov model and the DFAs. While this approach is heuristic and lacks any formal guarantees, it has
the basic elements common to machine improvisation schemes: (i) it involves randomly generating
strings from a formal language typically encoded as an automaton, (ii) it enforces diversity in the
generated strings, and (iii) it includes a requirement on which strings are admissible based on their
divergence from a reference string. The definition we propose below captures these elements in a
rigorous theoretical manner, suitable for further analysis.  In Sec.~\ref{section:finite-memory}, we
revisit the factor oracle, sketching how the notion of divergence from $\wref$ that it represents can be
encoded in our formalism.
 
 {
\setlength{\intextsep}{8pt}
\setlength{\belowcaptionskip}{-5pt}
\setlength{\abovecaptionskip}{0pt}
\begin{figure}[tb]
\centering
\begin{tikzpicture}[initial text=, transform shape, scale=0.8]

 \node[accepting, state, initial] (s0) {$s_0$}; 
 \node[accepting, state, right= of s0] (s1) {$s_1$}; 
 \node[accepting, state, right= of s1] (s2) {$s_2$};
 \node[accepting, state, right= of s2] (s3) {$s_3$}; 
 \node[accepting, state, right= of s3] (s4) {$s_4$};

 \path[->] 
 (s0) edge node [above] {$b$} (s1)
 (s1) edge node [above] {$b$} (s2)    
 (s2) edge node [above] {$a$} (s3)
 (s3) edge node [above] {$c$} (s4) 
 (s0) edge [bend left=40] node [above] {$a$} (s3)
 (s0) edge [bend left=50] node [above] {$c$} (s4)
 (s1) edge [bend left] node [above] {$a$} (s3)
 (s1) edge [bend left] node [above] {$\epsilon$} (s0) 
 (s2) edge [bend left ] node [above] {$\epsilon$} (s1) 
 (s3) edge [bend left] node [above] {$\epsilon$} (s0) 
 (s4) edge [bend left] node [above] {$\epsilon$} (s0); 

\end{tikzpicture}
\caption{Factor oracle constructed from the word $\wref = bbac$.}
\label{figure:factor-oracle}
\end{figure}
}
 
\subsection{Problem Definition}

We abbreviate deterministic and nondeterministic finite
automata as DFAs and NFAs respectively. 
We use the standard definition of probabilistic finite automata from
\cite{rabin-pfas}, where a string is accepted iff it causes the
automaton to reach an accepting state with probability greater than a
specified \emph{cut-point} $p \in [0, 1)$. 
We call a probabilistic finite automaton, together with a choice of
cut-point so that its language is definite, a PFA. 
We write $\Pr[ f(X) \suchthat X \leftarrow D ]$ for the probability of event
$f(X)$ given that the random variable $X$ is drawn from the
distribution $D$. 

\begin{definition}
An \emph{improvisation automaton} is a finite automaton (DFA, NFA, or
PFA) $\mathcal{\improvs}$ over a finite alphabet $\Sigma$. An
\emph{improvisation} is any word $w \in L(\mathcal{\improvs})$, and
$\improvs = L(\mathcal{\improvs})$ is the set of all improvisations. 
\label{defn:improv-auto}
\end{definition}

\begin{definition}
An \emph{admissibility predicate} is a computable predicate $\alpha :
\Sigma^* \rightarrow \{0,1\}$. An improvisation $w \in \improvs$ is
\emph{\creative} if $\alpha(w) = 1$. We write $\valids$ for the set of
all {\creative} improvisations. 
\label{defn:admiss}
\end{definition}

\begin{subparagraph}{Running Example.}
Our concepts will be illustrated with a simple 
example. Our aim is to produce variations of the binary string $s =
001$ of length 3, subject to the constraint that there cannot be two
consecutive 1s. So $\Sigma = \{0,1\}$, and 
$\mathcal{\improvs}$ is a DFA which accepts all length-3 strings 
that do not have two 1s in a row. To ensure that our
variations are similar to $s$, we let our admissibility predicate
$\alpha(w)$ be 1 if the Hamming distance between $w$ and $s$ is at
most 1, and 0 otherwise. Then the improvisations are the strings
$000$, $001$, $010$, $100$, and $101$, of which $000$, $001$, and
$101$ are admissible. \\
\end{subparagraph}

Intuitively, an improviser samples from the set of improvisations according to some
distribution. But what requirements must one impose on this distribution?  Since we want a variety
of improvisations, we require that each one is generated with probability at most some bound
$\rho$. By choosing a small value of $\rho$ we can thus ensure that many different improvisations
can be generated, and that no single one is output too frequently. Other constraints are possible,
e.g. requiring that every improvisation have nonzero probability, but we view this as too
restrictive: if there are a large number of possible improvisations, it should be acceptable for an
improviser to generate many but not all of them.
Another possibility would be to ensure variety by imposing some minimum distance between the improvisations.
This could be reasonable in a setting (such as music) where there is a natural metric on the space of improvisations, but we choose to keep our setting general and not assume such a metric.
Finally, we require our generated improvisation to
be admissible with probability at least $1-\epsilon$ for some specified $\epsilon$. When the
admissibility predicate encodes a notion of similarity to a reference string, for example, this allows us to
require that our improvisations usually be similar to the reference. Combining these requirements,
we obtain our definitions of an acceptable distribution over improvisations and thus of an
improviser:

\begin{definition}
Given $\mathcal{C} = (\mathcal{\improvs}, \alpha, \epsilon, \rho)$
with $\mathcal{\improvs}$ and $\alpha$ as in
Definitions~\ref{defn:improv-auto} and~\ref{defn:admiss}, $\epsilon \in [0,1]
\cap \Q$ an error probability, and $\rho \in (0,1] \cap \Q$ a
probability 
bound, a distribution $D : \Sigma^*
\rightarrow [0,1]$ with support $S$ is an
\emph{$(\epsilon,\rho)$-improvising distribution} if:
\begin{itemize}
\item $S \subseteq \improvs$
\item $\forall w \in S, D(w) \le \rho$
\item $\Pr[ w \in \valids \suchthat w \leftarrow D ] \ge 1 - \epsilon$
\end{itemize}
If there is an $(\epsilon,\rho)$-improvising distribution, we say that
$\mathcal{C}$ is $(\epsilon,\rho)$-\emph{feasible} (or simply {\em feasible}). An 
\emph{$(\epsilon,\rho)$-improviser} (or simply {\em improviser}) for a feasible $\mathcal{C}$ 
is an expected finite-time
probabilistic algorithm generating strings in $\Sigma^*$
whose output distribution (on empty
input) is an $(\epsilon,\rho)$-improvising distribution.
\label{defn:feasible}
\end{definition}

To summarize, if $\mathcal{C} = (\mathcal{\improvs}, \alpha, \epsilon, \rho)$ is feasible, there
exists a distribution satisfying the requirements in Definition~\ref{defn:feasible}, and an
improviser is a probabilistic algorithm for sampling from one.

\begin{subparagraph}{Running Example.}
For our running example, $\mathcal{C} = (\mathcal{\improvs}, \alpha,
0, 1/4)$ is not feasible since $\epsilon=0$ 
means we can only generate admissible improvisations, and since there
are only 3 of those we cannot possibly give them all probability at
most $1/4$. Increasing $\rho$ to $1/3$ would make $\mathcal{C}$
feasible. Increasing $\epsilon$ to $1/4$ would also work, allowing us
to return an inadmissible improvisation $1/4$ of the time:
an algorithm uniformly sampling from $\{000, 001, 101, 100\}$ would be
an improviser for $(\mathcal{\improvs}, \alpha, 1/4, 1/4)$. 
\end{subparagraph}

\begin{definition}
Given $\mathcal{C} = (\mathcal{\improvs}, \alpha, \epsilon, \rho)$,
the \emph{control improvisation (CI)} problem is to decide whether
$\mathcal{C}$ is feasible, and if so to generate an improviser for
$\mathcal{C}$. 
\end{definition}

Ideally, we would like an efficient algorithm to solve the CI problem. Furthermore,
the improvisers our algorithm produces should themselves be efficient, in the sense
that their runtimes are polynomial in the size of the original CI
instance. This leads to our last definition: 

\begin{definition}
A \emph{polynomial-time improvisation scheme} for a class
$\mathcal{P}$ of CI instances is a polynomial-time algorithm $S$
with the following properties: 
\begin{itemize}
\item for any $\mathcal{C} \in \mathcal{P}$, if $\mathcal{C}$ is
feasible then $S(\mathcal{C})$ is an improviser for $\mathcal{C}$, and
otherwise $S(\mathcal{C}) = \bot$  
\item there is a polynomial $p:\mathbb{R}\rightarrow \mathbb{R}$
such that if $G = S(\mathcal{C}) \ne \bot$, then $G$
has expected runtime at most $p(|\mathcal{C}|)$. 
\end{itemize}
\end{definition}

A polynomial-time improvisation scheme for a class of CI instances is
an efficient, uniform way to solve the control improvisation problem
for that class. In Sections \ref{section:finite-memory} and
\ref{section:complex-automata} we will investigate which classes have
such improvisation schemes. 

\section{Existence of Improvisers} \label{section:existence}

It turns out that the feasibility of an improvisation problem is completely determined by the sizes of $\improvs$ and $\valids$:

\begin{theorem} \label{theorem:feasibility}
For any $\mathcal{C} = (\mathcal{\improvs}, \alpha, \epsilon, \rho)$, the following are equivalent:
\begin{enumerate}[\quad (a)]
\item $\mathcal{C}$ is feasible.
\item $|\improvs| \ge 1/\rho$ and $|\valids| \ge (1-\epsilon)/\rho$.
\item There is an improviser for $\mathcal{C}$.
\end{enumerate}
\end{theorem}
\begin{proof}
\begin{description}
\item[(a)$\Rightarrow$(b):] Suppose $D$ is an $(\epsilon,\rho)$-improvising distribution with support $S$. Then $\rho |S| = \sum_{w \in S} \rho \ge \sum_{w \in S} D(w) = 1$, so $|\improvs| \ge |S| \ge 1/\rho$. We also have $\rho |S \cap \valids| = \sum_{w \in S \cap \valids} \rho \ge \sum_{w \in S \cap \valids} D(w) = \Pr[ w \in \valids \suchthat w \leftarrow D] \ge 1 - \epsilon$, so $|\valids| \ge |S \cap \valids| \ge (1-\epsilon)/\rho$.

\item[(b)$\Rightarrow$(c):] Defining $N = \ceil{(1-\epsilon)/\rho}$, we have $|\valids| \ge N$. If $N \ge 1/\rho$, then there is a subset $S \subseteq \valids$ with $|S| = \ceil{1/\rho}$. Since $1/\ceil{1/\rho} \le \rho$, the uniform distribution on $S$ is a $(0,\rho)$-improvising distribution. Since this distribution has finite support and rational probabilities, there is an expected finite-time probabilistic algorithm sampling from it, and this is a $(0,\rho)$-improviser. If instead $N < 1/\rho$, defining $M = \ceil{1/\rho} - N$ we have $M \ge 1$. Since $|\improvs| \ge \ceil{1/\rho} = N + M$, there are disjoint subsets $S \subseteq \valids$ and $T \subseteq \improvs$ with $|S| = N$ and $|T| = M$. Let $D$ be the distribution on $S \cup T$ where each element of $S$ has probability $\rho$ and each element of $T$ has probability $(1 - \rho N) / M = (1 - \rho N) / \ceil{(1/\rho) - N} = (1 - \rho N) / \ceil{(1 - \rho N) / \rho} \le \rho$. Then $\Pr[ w \in \valids | w \leftarrow D ] \ge \rho N \ge 1 - \epsilon$, so $D$ is a $(\epsilon,\rho)$-improvising distribution. As above there is an expected finite-time probabilistic algorithm sampling from $D$, and this is an $(\epsilon,\rho)$-improviser.

\item[(c)$\Rightarrow$(a):] Immediate. \qedhere
\end{description} 
\end{proof}

\begin{remark}
In fact, whenever $\mathcal{C}$ is feasible, the construction in the proof of Theorem \ref{theorem:feasibility} gives an improviser which works in nearly the most trivial possible way: it has two finite lists $S$ and $T$, flips a (biased) coin to decide which list to use, and then returns an element of that list uniformly at random.
\end{remark}

A consequence of this characterization is that when there are infinitely-many admissible improvisations, there is an improviser with zero error probability:

\begin{corollary}
If $\valids$ is infinite, $(\mathcal{\improvs}, \alpha, 0, \rho)$ is feasible for any $\rho \in (0,1] \cap \Q$.
\end{corollary}

In addition to giving conditions for feasibility, Theorem \ref{theorem:feasibility} yields an algorithm which is guaranteed to find an improviser for any feasible CI problem.

\begin{corollary}
If $\mathcal{C}$ is feasible, an improviser for $\mathcal{C}$ may be found by an effective procedure.
\end{corollary}
\begin{proof}
The sets $\improvs$ and $\valids$ are clearly computably enumerable, since $\alpha$ is computable. We enumerate $\improvs$ and $\valids$ until enough elements are found to perform the construction in Theorem \ref{theorem:feasibility}. Since $\mathcal{C}$ is feasible, the theorem ensures this search will terminate.
\end{proof}

We cannot give an upper bound on the time needed by this algorithm without knowing something about the admissibility predicate $\alpha$. Therefore although as noted in the remark above whenever there are improvisers at all there is one of a nearly-trivial form, actually finding such an improviser could be difficult. In fact, it could be faster to generate an improviser which is \emph{not} of this form, as seen for example in Sec.~\ref{section:finite-memory}.

\begin{corollary}
The set of feasible CI instances is computably enumerable but not computable.
\end{corollary}
\begin{proof}
Enumerability follows immediately from the previous Corollary. If checking whether $\mathcal{C}$ is feasible were decidable, then so would be checking if $|\valids| \ge (1-\epsilon)/\rho$, but this is undecidable since $\alpha$ can be an arbitrary computable predicate.
\end{proof}

\section{Finite-Memory Admissibility Predicates} \label{section:finite-memory}

In order to bound the time needed to find an improviser, we must
constrain the admissibility predicate $\alpha$. Perhaps the simplest
type of admissibility predicate is one which can be computed by a DFA,
i.e., one such that there is some DFA $\mathcal{D}$ which accepts a
word $w \in \Sigma^*$ iff $\alpha(w) = 1$. This captures the notion of
a \emph{finite-memory} admissibility predicate, where admissibility of
a word can be determined by scanning the word left-to-right, only
being able to remember a finite number of already-seen symbols. An
example of a finite-memory predicate $\alpha$ is one 
such that $\alpha(w) = 1$ iff each subword of $w$
of a fixed constant length satisfies some condition. By the
pumping lemma, such predicates have the property that continually
repeating some section of a word can produce an infinite family of
improvisations, which could be a disadvantage if looking for
``creative'', non-repetitive improvisations. However, in applications
such as music we impose a maximum length on improvisations, so this is
not an issue. 

\begin{example-indexed-by-theorem}[Factor Oracles]
Recall that one way of measuring the divergence of an improvisation
$w$ generated by the factor oracle $F$ built from a word $\wref$ is by
counting the number of non-direct transitions that $w$ causes $F$ to
take. Since DFAs cannot count without bound, we can use a sliding window of some
finite size $k$. Then our admissibility predicate $\alpha$ can be that
at any point as $F$ processes $w$, the number of the previous $k$
transitions which were non-direct lies in some interval $[\ell,h]$
with $0 \le \ell \le h \le k$. This predicate can be encoded as a DFA
of size $O(|F| \cdot 2^k)$ (see the Appendix for details). The size of
the automaton grows exponentially in the size of the window, but for
small windows it can be reasonable. 
\end{example-indexed-by-theorem}

When the admissibility predicate is finite-memory and the automaton
$\mathcal{\improvs}$ is a DFA, there is an efficient procedure to test
if an improviser exists and synthesize one if so. The construction is
similar to that of Theorem \ref{theorem:feasibility}, but avoids explicit enumeration of all
improvisations to be put in the range of the improviser. To avoid
enumeration we use a classic method of uniformly sampling from the
language of a DFA $\mathcal{D}$ (see for example
\cite{hickey-cohen,denise2006}). The next few lemmas summarize the
results we need, proofs being given in the Appendix for
completeness. The first step is to determine the size of the
language. 

\begin{restatable}{lemma-indexed-by-theorem}{lemmaDFACounting} \label{lemma-dfa-counting}
If $\mathcal{D}$ is a DFA, $|L(\mathcal{D})|$ can be computed in polynomial time.
\end{restatable}

Once we know the size of $L(\mathcal{D})$ we can efficiently sample
from it, handling infinite languages by sampling from a finite subset
of a desired size. 

\begin{restatable}{lemma-indexed-by-theorem}{lemmaDFAPumpSamp} \label{lemma-dfa-pump-samp}
There is a polynomial $p(x,y)$ such that for any $N \in \N$ and DFA $\mathcal{D}$ with infinite language, there is a probabilistic algorithm $S$ which uniformly samples from a subset of $L(\mathcal{D})$ of size $N$ in expected time at most $p(|\mathcal{D}|, \log N)$, and which can be constructed in the same time.
\end{restatable}

\begin{restatable}{lemma-indexed-by-theorem}{lemmaDFAUnifSamp} \label{lemma-dfa-unif-samp}
There is a polynomial $q(x)$ such that for any DFA $\mathcal{D}$ with finite language, there is a probabilistic algorithm $S$ which uniformly samples from $L(\mathcal{D})$ in expected time at most $q(|\mathcal{D}|)$, and which can be constructed in the same time.
\end{restatable}

Using these sampling techniques, we have the following:

\begin{theorem} \label{theorem-dfa-scheme}
The class of CI instances $\mathcal{C}$ where $\mathcal{\improvs}$ is
a DFA and $\alpha$ is computable by a DFA has a polynomial-time
improvisation scheme.  
\end{theorem}
\begin{proof}
The proof considers five cases. We first define some notation. 
Let $\mathcal{D}$ denote the DFA giving $\alpha$. Letting
$\mathcal{\valids}$ be the product of $\mathcal{\improvs}$
and $\mathcal{D}$, we have $\valids = L(\mathcal{\valids})$. This
product can be computed in polynomial time since the automata are both
DFAs, and $|\mathcal{\valids}|$ is polynomial in $|\mathcal{C}|$ and
$|\mathcal{D}|$. In some of the cases below we will also use a DFA
$\mathcal{B}$ which is the synchronous product of $\mathcal{\improvs}$
and the complement of $\mathcal{\valids}$. Clearly $L(\mathcal{B}) =
\improvs \setminus \valids$, and the size of $\mathcal{B}$ and the
time needed to construct it are also polynomial in $|\mathcal{C}|$ and
$|\mathcal{D}|$. 

Next we compute $|\valids| = |L(\mathcal{\valids})|$ and $| \improvs |
= |L(\mathcal{\improvs})|$ in polynomial time using Lemma
\ref{lemma-dfa-counting}. There are now several cases (illustrated in Figure
\ref{figure-proof-cases}): 

{
\setlength{\intextsep}{8pt}
\setlength{\belowcaptionskip}{-5pt}
\setlength{\abovecaptionskip}{2pt}
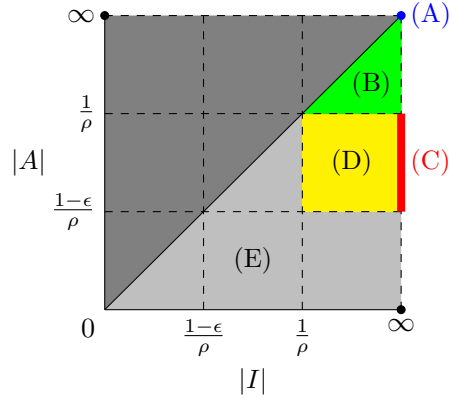
\begin{figure}
\centering
\begin{tikzpicture}[scale=1.3]

\path [fill=gray] (0,0) -- (0,3) -- (3,3) -- (0,0); 
\path [fill=lightgray] (0,0) -- (2,2) -- (2,1) -- (3,1) -- (3,0) -- (0,0); 
\path [fill=yellow] (2,2) -- (3,2) -- (3,1) -- (2,1) -- (2,2); 
\draw [green, fill=green] (2,2) -- (3,3) -- (3,2) -- (2,2); 

\draw (0,0) -- (3,0) node[anchor=north] {\large $\infty$};
\draw (0,0) node[anchor=north east] {\large $0$}
	(1,0) node[anchor=north] {\large $\frac{1-\epsilon}{\rho}$}
	(2,0) node[anchor=north] {\large $\frac{1}{\rho}$};
\draw (1.5, -0.5) node[anchor=north] {\normalsize $| \improvs |$};

\draw (0,0) -- (0,3) node[anchor=east] {\large $\infty$};
\draw (0,1) node[anchor=east] {\large $\frac{1-\epsilon}{\rho}$}
	(0,2) node[anchor=east] {\large $\frac{1}{\rho}$};
\draw (-0.5, 1.5) node[anchor=east] {\normalsize $| \valids |$};

\draw [fill] (0,3) circle (0.04)
	(3,0) circle (0.04);

\draw (0,0) -- (3,3); 
\draw [dashed] (1,0) -- (1,3);
\draw [dashed] (2,0) -- (2,3);
\draw [dashed] (3,0) -- (3,3);
\draw [dashed] (0,1) -- (3,1);
\draw [dashed] (0,2) -- (3,2);
\draw [dashed] (0,3) -- (3,3);

\draw [red, line width=3] (3,1) -- (3,2); 
\draw [blue, fill=blue] (3,3) circle (0.04); 

\node at (1.5,0.5) {\normalsize (E)};
\node at (3,3) [blue, anchor=west] {\normalsize (A)};
\node at (2.7,2.3) {\normalsize (B)};
\node at (3,1.5) [red, anchor=west] {\normalsize (C)};
\node at (2.5,1.5) {\normalsize (D)};

\end{tikzpicture}
\caption{Cases for Theorem \ref{theorem-dfa-scheme}. The dark gray
region cannot occur.}
\label{figure-proof-cases}
\end{figure}
}

\begin{enumerate}[(A)]
\item \label{case:pump-v} $|\valids| = \infty$: Applying Lemma
\ref{lemma-dfa-pump-samp} to $\mathcal{\valids}$ with $N =
\ceil{1/\rho}$, we obtain a probabilistic algorithm $S$ which uniformly samples from a
subset of $L(\mathcal{\valids}) = \valids$ of size
$\ceil{1/\rho}$. Since $1/\ceil{1/\rho} \le \rho$, we have that $S$ is
a $(0,\rho)$-improviser and return it. 

\item \label{case:unif-v} $1/\rho \le |\valids| < \infty$: Applying Lemma \ref{lemma-dfa-unif-samp} to $\mathcal{\valids}$, we obtain a probabilistic algorithm $S$ which uniformly samples from $L(\mathcal{\valids}) = \valids$. Since $1/|\valids| \le \rho$, we have that $S$ is a $(0,\rho)$-improviser and return it.

\item \label{case:pump-i-unif-v} $(1-\epsilon)/\rho \le |\valids| < 1/\rho$ and $| \improvs | = \infty$: Applying Lemma \ref{lemma-dfa-unif-samp} to $\mathcal{\valids}$ we obtain $S$ as in the previous case. Defining $M = \ceil{1/\rho} - |\valids|$, we have $\infty = |L(\mathcal{B})| > M \ge 1$. Applying Lemma \ref{lemma-dfa-pump-samp} to $\mathcal{B}$ with $N = M$ yields a probabilistic algorithm $S'$ which uniformly samples from a subset of $L(\mathcal{B}) = \improvs \setminus \valids$ of size $M$. Let $G$ be a probabilistic algorithm which with probability $\rho |\valids|$ executes $S$, and otherwise executes $S'$. Then since $L(\mathcal{\valids}) = \valids$ and $L(\mathcal{B}) = \improvs \setminus \valids$ are disjoint, every word generated by $G$ has probability either $(\rho |\valids|) / |\valids| = \rho$ (if it is in $\valids$) or $(1 - \rho |\valids|) / M = (1 - \rho |\valids|) / \ceil{(1 - \rho |\valids|) / \rho} \le \rho$ (if it is in $\improvs \setminus \valids$). Also, $G$ outputs a member of $\valids$ with probability $\rho |\valids| \ge 1 - \epsilon$, so $G$ is an $(\epsilon,\rho)$-improviser and we return it.

\item \label{case:unif-i-unif-v} $(1-\epsilon)/\rho \le |\valids| < 1/\rho \le | \improvs | < \infty$: As in the previous case, except obtaining $S'$ by applying Lemma \ref{lemma-dfa-unif-samp} to $\mathcal{B}$. Since $| \improvs | \ge \ceil{1/\rho}$, we have $|L(\mathcal{B})| = | \improvs \setminus \valids | \ge M$ and so $G$ as constructed above is an $(\epsilon, \rho)$-improviser.

\item \label{case:infeasible} $| \improvs | < 1/\rho$ or $|\valids| < (1-\epsilon) / \rho$: By Theorem \ref{theorem:feasibility}, $\mathcal{C}$ is not feasible, so we return $\bot$.
\end{enumerate}

This procedure takes time polynomial in $|\mathcal{\improvs}|$, $|\mathcal{D}|$, and $\log(1/\rho)$, so it is polynomial-time. Also, a fixed polynomial in these quantities bounds the expected runtime of the generated improviser, so the procedure is a polynomial-time improvisation scheme.
\end{proof}

\begin{subparagraph}{Running Example.}
Recall that for our running example $\mathcal{C} = (\mathcal{\improvs}, \alpha, 1/4, 1/4)$, we have $\improvs = \{000, 001, 010, 100, 101\}$ and $\valids = \{000, 001, 101\}$. 
Since $|\valids| = 3$ and $|\improvs| = 5$, we are in case (\ref{case:unif-i-unif-v}) of Theorem \ref{theorem-dfa-scheme}.
So our scheme uses Lemma \ref{lemma-dfa-unif-samp} to obtain $S$ and $S'$ uniformly sampling from $\valids$ and $\improvs \setminus \valids = \{010, 100\}$ respectively.
It returns a probabilistic algorithm $G$ that executes $S$ with probability $\rho |\valids| = 3/4$ and otherwise executes $S'$.
So $G$ returns $000$, $001$, and $101$ with probability $1/4$ each, and $010$ and $100$ with probability $1/8$ each.
The output distribution of $G$ satisfies our conditions, so it is an improviser for $\mathcal{C}$.
\end{subparagraph}

\section{More Complex Automata} \label{section:complex-automata}

While counting the language of a DFA is easy, in the case of an NFA it
is much more difficult, and so there are unlikely to be
polynomial-time improvisation schemes for more complex automata. Let
$\mathcal{N}_1$ and $\mathcal{N}_2$ be the classes of CI instances
where $\mathcal{\improvs}$ or $\alpha$ respectively are given by an
NFA, and the other is given by a DFA. Then denoting by $\mathcal{N}$ 
either of these classes, we have (deferring full proofs from this section to the Appendix): 

\begin{restatable}{theorem}{theoremNFAHardness} \label{theorem:nfa-hardness}
Determining whether $\mathcal{C} \in \mathcal{N}$ is feasible is $\sharpP$-hard.
\end{restatable}
\begin{proof}[Proof sketch]
The problem of counting the language of an NFA, which is
$\sharpP$-hard \cite{sharpNFA}, is polynomially reducible to that of
checking if $\mathcal{C} \in \mathcal{N}$ is feasible.
\end{proof}
\begin{remark}
Determining feasibility of $\mathcal{N}$-instances is not a counting
problem, so it is not $\sharpP$-complete, but it is clearly in
$\P^\sharpP$: we construct the automata $\mathcal{\improvs}$ and
$\mathcal{\valids}$ as in Theorem \ref{theorem-dfa-scheme} (now they
can be NFAs), count their languages using $\sharpP$, and apply Theorem
\ref{theorem:feasibility}. 
\end{remark}

\begin{corollary}
If there is a polynomial-time improvisation scheme for $\mathcal{N}$, then $\P = \P^\sharpP$.
\end{corollary}

This result indicates that in general, the control improvisation
problem is probably intractable in the presence of NFAs. Some special
cases could still be handled in practice: for example, if the NFA is
very small it could be converted to a DFA. Another tractable case is
where although one of $\mathcal{\improvs}$ or $\mathcal{\valids}$ (as
in Theorem \ref{theorem-dfa-scheme}) is an NFA, it has infinite
language (this can clearly be detected in polynomial time). If
$\mathcal{\valids}$ is an NFA with infinite language we can use case
(\ref{case:pump-v}) of the proof of Theorem \ref{theorem-dfa-scheme}, since an NFA
can be pumped in the same way as a DFA. If instead $\mathcal{\valids}$
is a DFA with finite language but $\mathcal{\improvs}$ is an NFA with
infinite language, one of the other cases (\ref{case:unif-v}),
(\ref{case:pump-i-unif-v}), or (\ref{case:infeasible}) applies, and in
case (\ref{case:pump-i-unif-v}) we can sample $\improvs \setminus
\valids$ by pumping $\mathcal{\improvs}$ enough to ensure we get a
string longer than any accepted by $\mathcal{\valids}$. Table
\ref{table:complexities} in Section \ref{section:conclusion}
summarizes these cases. 

For still more complex automata, the CI problem becomes even
harder. In fact, it is impossible if we allow either
$\mathcal{\improvs}$ or $\alpha$ to be given by a PFA. Let
$\mathcal{P}_1$ and $\mathcal{P}_2$ be the classes of CI instances
where each of these respectively are given by a PFA, and the other is
given by a DFA. Then letting $\mathcal{P}$ be either of these classes,
we have: 

\begin{restatable}{theorem}{theoremPFAHardness} \label{theorem:pfa-hardness}
Determining whether $\mathcal{C} \in \mathcal{P}$ is feasible is undecidable.
\end{restatable}
\begin{proof}[Proof sketch]
Checking the feasibility of $\mathcal{C}$ amounts to counting the language of a PFA, but determining whether the language of a PFA is empty is undecidable \cite{nasu-honda,condon-lipton}.
\end{proof}

\section{Symbolic Techniques} \label{section:symbolic}

Previously we have assumed that the automata defining a control
improvisation problem were given explicitly. However, in practice
there may be insufficient memory to store full transition tables, in
which case an implicit representation is required. This prevents us from using
the polynomial-time improvisation scheme of Theorem
\ref{theorem-dfa-scheme}, so we must look for alternate methods. These
will depend on the type of implicit representation used. We focus on
representations of DFAs and NFAs by propositional formulae, as used
for example in bounded model checking \cite{bmc}. 
\begin{definition}
A \emph{symbolic automaton} is a transition system over states $S \subseteq \{0,1\}^n$ and inputs $\Sigma \subseteq \{0,1\}^m$ represented by:
\begin{itemize}
\item a formula $\mathrm{init}(\overline{x})$ which is true iff $\overline{x} \in \{0,1\}^n$ is an initial state,
\item a formula $\mathrm{acc}(\overline{x})$ which is true iff $\overline{x} \in \{0,1\}^n$ is an accepting state, and
\item a formula $\delta(\overline{x}, \overline{a}, \overline{y})$ which is true iff there is a transition from $\overline{x} \in \{0,1\}^n$ to $\overline{y} \in \{0,1\}^n$ on input $\overline{a} \in \{0,1\}^m$ .
\end{itemize}
A symbolic automaton accepts words in $\Sigma^*$ according to the usual definition for NFAs.
\end{definition}

Given a symbolic automaton, it is straightforward to generate a
formula whose models correspond, for example, to accepting paths of at
most a given length (see \cite{bmc} for details). A SAT solver can
then be used to find such a path. We refer to the length of the
longest simple accepting path as the \emph{diameter} of the
automaton. This will be an important parameter in the runtime of our
algorithms. In some cases an upper bound on the diameter is known
ahead of time: for example, if we only want improvisations of up to
some maximum length, and have encoded that constraint in
$\mathcal{\improvs}$. If the diameter is not known, it can be found
iteratively with SAT queries asserting the existence of a simple
accepting path of length $n$, increasing $n$ until we find no such
path exists. The diameter could be exponentially large compared to the
symbolic representation, but this is a worst-case scenario. 

Our approach for solving the control improvisation problem with symbolic automata will be to adapt the procedure of Theorem \ref{theorem-dfa-scheme}, replacing the counting and sampling techniques used there with ones that work on symbolic automata. For language size estimation we use the following:
\begin{restatable}{lemma-indexed-by-theorem}{lemmaSymbolicCount} \label{lemma:symbolic-count}
If $\mathcal{S}$ is a symbolic automaton with diameter $D$, for any $\tau, \delta > 0$ we can compute an estimate of $|L(\mathcal{S})|$ accurate to within a factor of $1+\tau$ in time polynomial in $|\mathcal{S}|$, $D$, $1/\tau$, and $\log(1/\delta)$ relative to an {\NP} oracle.
\end{restatable}

Sampling from infinite languages can be done by a direct adaptation of the method for explicit DFAs in Lemma \ref{lemma-dfa-pump-samp}.
\begin{restatable}{lemma-indexed-by-theorem}{lemmaSymbPumpSamp} \label{lemma:symb-pump-samp}
There is a polynomial $p(x,y,z)$ such that for any $N \in \N$ and symbolic automaton $\mathcal{Y}$ with infinite language and diameter $D$, there is a probabilistic oracle algorithm $S^{\NP}$ which uniformly samples from a subset of $L(\mathcal{Y})$ of size $N$ in expected time at most $p(|\mathcal{Y}|, D, \log N)$ and which can be constructed in the same time.
\end{restatable}

To sample from a finite language, we use techniques for almost-uniform
generation of models of propositional formulae. In theory uniform
sampling can be done exactly using a SAT solver \cite{bgp}, but the
only algorithms which work in practice are \emph{approximate} uniform
generators such as {\UniGen} \cite{unigen}. This algorithm guarantees
that the probability of returning any given model is within a factor
of $1 + \tau$ of the uniform probability, for any given $\tau > 6.84$
(the constant is for technical reasons specific to
{\UniGen}). {\UniGen} can also do projection sampling, i.e., sampling
where two models are considered identical if they agree on the set of
variables being projected onto. Henceforth, for simplicity, we will
assume we have a generic almost-uniform generator that can do
projection, and will ignore the $\tau > 6.84$ restriction imposed by
{\UniGen} (although we might want to abide by this in practice in
order to be able to use the fastest available algorithm). We assume
that the generator runs in time polynomial in $1/\tau$ and the size of
the given formula relative to an {\NP} oracle, and succeeds with at
least some fixed constant probability. 
\begin{restatable}{lemma-indexed-by-theorem}{lemmaSymbUnifSamp} \label{lemma:symb-unif-samp}
There is a polynomial $q(x,y,z)$ such that for any $\tau > 0$ and symbolic automaton $\mathcal{Y}$ with finite language and diameter $D$, there is a probabilistic oracle algorithm $S^{\NP}$ which samples from $L(\mathcal{S})$ uniformly up to a factor of $1+\tau$ in expected time at most $q(|\mathcal{Y}|, D, 1/\tau)$, and which can be constructed in the same time.
\end{restatable}

Now we can put these methods together to get a version of Theorem \ref{theorem-dfa-scheme} for symbolic automata. The major differences are that this scheme requires an {\NP} oracle, has some probability of failure (which can be specified), and returns an improviser with a slightly sub-optimal value of $\rho$. The proof generally follows that of Theorem \ref{theorem-dfa-scheme}, so we only sketch the differences here (see the Appendix for a full proof).
\begin{restatable}{theorem}{theoremSymbolicScheme}
There is a procedure that given any CI problem $\mathcal{C}$ where $\mathcal{\improvs}$ and $\alpha$ are given by symbolic automata with diameter at most $D$, and any $\epsilon \in [0,1]$, $\rho \in (0,1]$, and $\tau, \delta > 0$, if $\mathcal{C}$ is $(\epsilon, \rho/(1+\tau))$-feasible returns an $(\epsilon, (1+\tau)^2 (1+\epsilon) \rho)$-improviser with probability at least $1-\delta$. Furthermore, the procedure and the improvisers it generates run in expected time given by some fixed polynomial in $|\mathcal{C}|$, $D$, $1/\tau$, and $\log(1/\delta)$ relative to an {\NP} oracle.
\end{restatable}
\begin{proof}[Proof sketch]
We first compute estimates $E_{\valids}$ and $E_{\improvs}$ of $|\valids|$ and $|\improvs|$ respectively using Lemma \ref{lemma:symbolic-count}. Then we break into cases as in Theorem \ref{theorem-dfa-scheme}:
\begin{enumerate}[(A)]
\item $E_{\valids} = \infty$: As in case (\ref{case:pump-v}) of Theorem \ref{theorem-dfa-scheme}, using Lemma \ref{lemma:symb-pump-samp} in place of Lemma \ref{lemma-dfa-pump-samp}. We obtain a $(0, \rho)$-improviser.

\item $1/\rho \le E_{\valids} < \infty$: As in case (\ref{case:unif-v}) of Theorem \ref{theorem-dfa-scheme}, using Lemma \ref{lemma:symb-unif-samp} in place of Lemma \ref{lemma-dfa-unif-samp}. Since we can do only approximate counting and sampling, we obtain a $(0, (1+\tau)^2 \rho)$-improviser.

\item $(1-\epsilon)/\rho \le E_{\valids} < 1/\rho$ and $E_{\improvs} = \infty$: As in case (\ref{case:pump-i-unif-v}) of Theorem \ref{theorem-dfa-scheme}, using Lemmas \ref{lemma:symb-pump-samp} and \ref{lemma:symb-unif-samp} in place of Lemmas \ref{lemma-dfa-pump-samp} and \ref{lemma-dfa-unif-samp}. Our use of approximate counting/sampling means we obtain only an $(\epsilon, (1+\tau)^2 \rho)$-improviser.

\item $(1-\epsilon)/\rho \le E_{\valids} < 1/\rho \le E_{\improvs} < \infty$: We cannot use the procedure in case (\ref{case:unif-i-unif-v}) of Theorem \ref{theorem-dfa-scheme}, since it may generate an element of $L(\mathcal{B})$ with too high probability if our estimate $E_{\valids}$ is sufficiently small. Instead we sample almost-uniformly from $L(\mathcal{\valids})$ with probability $\epsilon$, and from $L(\mathcal{\improvs})$ with probability $1-\epsilon$. This yields an $(\epsilon, (1+\tau)^2 (1+\epsilon) \rho)$-improviser.

\item \label{case:symb-infeasible} $E_{\improvs} < 1/\rho$ or $E_{\valids} < (1-\epsilon) / \rho$: We return $\bot$.
\end{enumerate}

If $\mathcal{C}$ is $(\epsilon, \rho/(1+\tau))$-feasible, case (\ref{case:symb-infeasible}) happens with probability less than $\delta$ by Theorem \ref{theorem:feasibility}. Otherwise, we obtain an $(\epsilon, (1+\tau)^2 (1+\epsilon) \rho)$-improviser.
\end{proof}

Therefore, it is possible to {\em approximately} solve the control
improvisation problem when the automata are given by a succinct
propositional formula representation. This allows working with general
NFAs, and very large automata that cannot be stored explicitly, but
comes at the cost of using a SAT solver (perhaps not a heavy cost
given the dramatic advances in the capacity of SAT solvers) and
possibly having to increase $\rho$ by a small factor. 

\section{Conclusion} \label{section:conclusion}

In this paper, we introduced control improvisation, the problem of
creating improvisers that randomly generate variants of words in the
languages of automata. We gave precise conditions for when improvisers
exist, and investigated the complexity of finding improvisers for
several major classes of automata. In particular, we showed that the
control improvisation problem for DFAs can be solved in polynomial
time, while it is intractable in most cases for NFAs and undecidable 
for PFAs. These results are summarized in Table
\ref{table:complexities}. Finally, we studied the case where the
automata are presented symbolically instead of explicitly, and showed
that the control improvisation problem can still be solved
approximately using SAT solvers. 

One interesting direction for future
work would be to find other tractable cases of the control
improvisation problem deriving from finer structural properties of the
automata than just determinism. Extensions of the theory to other classes of
formal languages, for instance context-free languages represented by 
pushdown automata or context-free grammars, are also worthy of study. 
Finally, we are investigating
further applications, particularly in the areas of testing, security,
and privacy. 

\setlength{\tabcolsep}{5pt}
\renewcommand{\arraystretch}{1.2}
\begin{table}[tb]
\begin{center}
\begin{tabular}{cr||c|c|c|c|}
 & $\alpha$ & \textbf{DFA} & \multicolumn{2}{|c|}{\textbf{NFA}} & \textbf{PFA} \\
$\mathcal{\improvs}$ &  &  & $L(\mathcal{\valids}) = \infty$ & $L(\mathcal{\valids}) < \infty$ &  \\
\hline
\hline
\textbf{DFA} &  & \multicolumn{2}{c|}{\multirow{2}{*}{poly-time}} & \multirow{3}{*}{\sharpP-hard} & \multirow{4}{*}{} \\
\cline{1-2}
\multirow{2}{*}{\textbf{NFA}} & $L(\mathcal{\improvs}) = \infty$ & \multicolumn{2}{c|}{} & &  \\
\cline{2-4}
& $L(\mathcal{\improvs}) < \infty$ & \sharpP-hard & - &   &  \\
\cline{1-5}
\textbf{PFA} &  & \multicolumn{4}{r|}{undecidable} \\
\hline \\
\end{tabular}
\caption{Complexity of the control improvisation problem when $\mathcal{\improvs}$ and $\alpha$ are given by various different types of automata. The cell marked `-' is impossible since $L(\mathcal{\valids}) \subseteq L(\mathcal{\improvs})$.}
\label{table:complexities}
\end{center}
\label{defaulttable}
\end{table}

\subparagraph*{Acknowledgements}
The first three authors dedicate this paper to the memory of the fourth author, David Wessel, who passed away while it was being written.
We would also like to thank Ben Caulfield, Orna Kupferman, Markus Rabe, and the anonymous reviewers for their helpful comments.
This work is supported in part by the National Science Foundation Graduate Research Fellowship Program 
under Grant No.~DGE-1106400, by the NSF Expeditions grant CCF-1139138, and 
by TerraSwarm, one of six centers of STARnet, a Semiconductor Research Corporation program sponsored by MARCO and DARPA.

\bibliographystyle{plain}
\bibliography{main}

\appendix

\section{Proofs}

\setcounter{section}{4}
\begin{example-indexed-by-theorem}[Factor Oracles]
The factor oracle-based admissibility predicate $\alpha$ described above can be encoded as a DFA of size $O(|F| \cdot 2^k)$ as follows: we have a copy of $F$, denoted $F_s$, for every string $s \in \{0,1\}^k$, each bit of $s$ indicating whether the corresponding previous transition (out of the last $k$) was non-direct. As each new symbol is processed, we execute the current copy of $F$ as usual, but move to the appropriate state of the copy of $F$ corresponding to the new $k$-transition history, i.e., if we were in $F_s$, we move to $F_t$ where $t$ consists of the last $k-1$ bits of $s$ followed by a 0 if the transition we took was direct and a 1 otherwise. Making the states of $F_s$ accepting iff the number of 1s in $s$ is in $[\ell,h]$, this automaton represents $\alpha$ as desired.
\end{example-indexed-by-theorem}

\lemmaDFACounting*
\begin{proof}
First we prune irrelevant states unreachable from the initial state or from which no accepting state can be reached (this pruning can clearly be done in polynomial time). If the resulting graph contains a cycle (also detectable in polynomial time), we return $\infty$. Otherwise $\mathcal{D}$ is a DAG with multiple edges, and every sink is an accepting state. For each accepting state $s$ we add a new vertex and an edge to it from $s$. Then there is a one-to-one correspondence between accepting words of $\mathcal{D}$ and paths from the initial state to a sink. Now we can compute for each vertex $v$ the number of paths $p_v$ from it to a sink using the usual linear-time DAG algorithm (traversal in reverse topological order) modified slightly to handle multiple edges. We return $p_v$ with $v$ the initial state.
\end{proof}

\lemmaDFAPumpSamp*
\begin{proof}
Having pruned $\mathcal{D}$ as in Lemma \ref{lemma-dfa-counting}, since $L(\mathcal{D})$ is infinite there must be some state $s$ of $\mathcal{D}$ such that
\begin{itemize}
\item there is a word $x \in \Sigma^*$ which takes $\mathcal{D}$ from its initial state to $s$,
\item there is a nonempty word $y \in \Sigma^*$ which takes $\mathcal{D}$ from $s$ to itself, and
\item there is a word $z \in \Sigma^*$ which takes $\mathcal{D}$ from $s$ to an accepting state.
\end{itemize}
We can find $x,y,z \in \Sigma^*$ as above with $|x|,|y|,|z| \le |\mathcal{D}|$, in time polynomial in $|\mathcal{D}|$. Then we have $x y^n z \in L(\mathcal{D})$ for any $n \in \N$. We form a probabilistic algorithm $S$ which acts as follows: it prints $x$, then picks an integer uniformly at random from $[0,N-1]$ and prints that many copies of $y$, before finally printing $z$. Clearly the output of $S$ is a uniform sample from a subset of $L(\mathcal{D})$ of size $N$. Constructing $S$ takes time polynomial in $|\mathcal{D}|$ (as this bounds the sizes of $x$, $y$, and $z$) and $\log N$, and $S$ runs in expected time bounded by a fixed polynomial in these values.
\end{proof}

\lemmaDFAUnifSamp*
\begin{proof}
Prune $\mathcal{D}$ and compute the path counts $p_v$ as in Lemma \ref{lemma-dfa-counting}. To every edge $(u,v)$ in $\mathcal{D}$ assign the weight $p_v / p_u$. It is clear that at every vertex the sum of the weights of the outgoing edges is 1 (unless the vertex is a sink). We prove by induction along reverse topological order that treating these weights as transition probabilities, starting from any state $u$ and talking a random walk until a sink is reached we obtain a uniform distribution over all paths from $u$ to a sink. If $u$ is a sink this holds trivially. If $u$ has a nonempty set of children $S$, then by the inductive hypothesis for every $v \in S$ starting a walk at $v$ gives a uniform distribution over the $p_v$ paths from $v$ to a sink. Therefore the probability of following any such path starting at $u$ is $(p_v/p_u) \cdot (1/p_v) = 1/p_u$. So the result holds by induction. In particular, if we start from the initial state we obtain a uniform distribution over all paths to a sink, and thus a uniform distribution over $L(\mathcal{D})$. Since all probabilities are rational with denominators bounded by ${|\Sigma|}^{|\mathcal{D}|}$, this walk can be performed by a probabilistic algorithm $S$ of size polynomial in $|\mathcal{D}|$, with expected time bounded by a fixed polynomial in $|\mathcal{D}|$. Then $S$ returns a uniform sample from $L(\mathcal{D})$, and it can be constructed in time polynomial in $|\mathcal{D}|$.
\end{proof}

\theoremNFAHardness*
\begin{proof}
We prove this for $\mathcal{N}_1$ --- the other case is analogous. As shown in \cite{sharpNFA}, the problem of determining $| L(\mathcal{M}) \cap \Sigma^m |$ given an NFA $\mathcal{M}$ over an alphabet $\Sigma$ and $m \in \N$ in unary is $\sharpP$-complete. We give a polynomial-time (Cook) reduction from this problem to checking feasibility of a CI instance in $\mathcal{N}_1$.

As noted in \cite{sharpNFA}, we can in polynomial time (in $m$, which is acceptable since $m$ is given in unary) construct an NFA $\mathcal{M'}$ such that $|L(\mathcal{M'})| = |L(\mathcal{M}) \cap \Sigma^m|$. Then we construct the trivial DFA $\mathcal{T}$ accepting all of $\Sigma^*$, and consider the CI instances $\mathcal{C}_N = (\mathcal{M'}, \mathcal{T}, 0, 1/N)$. Clearly for these instances we have $\improvs = \valids = L(\mathcal{M'})$. By Theorem \ref{theorem:feasibility}, $\mathcal{C}_N$ is feasible iff $|L(\mathcal{M}) \cap \Sigma^m| = |L(\mathcal{M'})| \ge N$, and since $\mathcal{C}_N \in \mathcal{N}_1$ for any $N \in \N$ we can determine whether this is the case. Since $|L(\mathcal{M}) \cap \Sigma^m| \le |\Sigma|^m$, using binary search we can find the exact value of $|L(\mathcal{M}) \cap \Sigma^m|$ with polynomially-many such queries.
\end{proof}

\theoremPFAHardness*
\begin{proof}
We prove this for $\mathcal{P}_1$, the other case being similar. Given a PFA $\mathcal{A}$ with cut-point $p$ over an alphabet $\Sigma$ with at least two symbols, determining whether the language of $\mathcal{A}$ is empty (i.e. whether it accepts no words with probability greater than $p$) is undecidable \cite{nasu-honda,condon-lipton}. For any $N > 0$, we can construct a PFA $\mathcal{A'}$ by adding new states and deterministic transitions to $\mathcal{A}$ so that there are exactly $N$ words taking $\mathcal{A'}$ from its initial state to the initial state of $\mathcal{A}$, and any word which does not have one of these as a prefix causes $\mathcal{A'}$ to reject with probability 1. Then $L(\mathcal{A'})$ is $L(\mathcal{A})$ with one of the $N$ prefixes added to each word. Therefore $|L(\mathcal{A'})| \ge N$ iff $L(\mathcal{A}) \ne \emptyset$, and so it is undecidable to determine whether the language of a PFA has at least $N$ elements.

Constructing the trivial DFA $\mathcal{T}$ accepting all of $\Sigma^*$, the CI instance $\mathcal{C} = (\mathcal{A}, \mathcal{T}, 0, 1/N)$ satisfies $\improvs = \valids = L(\mathcal{A})$. By Theorem \ref{theorem:feasibility}, $\mathcal{C}$ is feasible iff $|L(\mathcal{A})| \ge N$. Since checking this latter condition is undecidable, so is determining feasibility of $\mathcal{P}_1$-instances.
\end{proof}

\lemmaSymbolicCount*
\begin{proof}
Recall that the algorithm in Lemma \ref{lemma-dfa-counting} detected
accepting cycles by finding words $x,y,z$ taking the automaton to some
state $s$, from $s$ to $s$ along at least one transition, and to an
accepting state respectively. If we impose the additional constraint
$|x|, |y|, |z| \le n$, then we can check the existence of such words
with a single SAT query. So to test whether $L(\mathcal{S})$ is
infinite, we use one query for each $n \le D$: since $D$ is the
diameter, we are guaranteed to find an accepting cycle if one
exists. Thus if any query is satisfiable, we return $\infty$. 

If instead all the queries fail, then all words in $L(\mathcal{S})$
have length at most $D$. So the SAT query $\phi$ for accepting paths
of length at most $D$ in fact matches every accepting path. Models of
this formula then correspond to accepting words if we project onto the
input variables (i.e. ignore the values of the variables which encode
states). Therefore $|L(\mathcal{S})|$ is equal to the number of models
of $\phi$ after projection. The general problem of counting models of
a propositional formula (even without projection) is \sharpP-complete,
but using the SAT solver we can get probabilistic bounds. An
approximate model counter such as {\ApproxMC} \cite{approxmc} can
return an estimate of the number of models of $\phi$ which is accurate
to within a factor of $1 + \tau$ with probability at least $1 -
\delta$. In fact {\ApproxMC} can be easily modified to do projection
counting (see \cite{unigen}), giving us the required estimate of
$|L(\mathcal{S})|$. 

The first stage of this process clearly takes time polynomial in $|\mathcal{S}|$ and $D$ relative to the oracle. {\ApproxMC} runs in time polynomial in $|\phi| = O(D |\mathcal{S}|)$, $1/\tau$, and $\log(1/\delta)$ relative to the oracle, so this procedure does as well.
\end{proof}

\lemmaSymbPumpSamp*
\begin{proof}
We look for an accepting cycle using the method in Lemma \ref{lemma:symbolic-count}. One will be found since $|L(\mathcal{Y})|$ is infinite, and then we can pump it to get $N$ different words just as in Lemma \ref{lemma-dfa-pump-samp}.
\end{proof}

\lemmaSymbUnifSamp*
\begin{proof}
As noted in Lemma \ref{lemma:symbolic-count}, since the language is finite every word in it has length at most $D$. Constructing the formula $\phi$ from that lemma, once we project onto the input variables there is a one-to-one correspondence between accepting words and models of $\phi$. So we need to almost-uniformly generate projected models of a propositional formula. We use an almost-uniform generator as described above, whose runtime will be polynomial in $|\phi| = O(D|\mathcal{Y}|)$ and $1/\tau$ relative to the oracle.
\end{proof}

\theoremSymbolicScheme*
\begin{proof}
The procedure begins by deriving the symbolic representations of $\mathcal{\valids}$ and $\mathcal{B}$ (the product of $\mathcal{\improvs}$ and the complement of $\mathcal{\valids}$). Next we estimate $|\valids| = |L(\valids)|$ and $|\improvs| = |L(\mathcal{\improvs})|$ using Lemma \ref{lemma:symbolic-count} with a confidence of $(1-\delta)^{1/2}$. Then with probability at least $1 - \delta$, both these estimates are within a factor of $1+\tau$ of the true values. We assume this is the case for the rest of the proof, making no guarantees otherwise. We now break into the same cases as Theorem \ref{theorem-dfa-scheme}, using our estimates $E_{\valids}$ and $E_{\improvs}$ of $|\valids|$ and $|\improvs|$ respectively.

\begin{enumerate}[(A)]
\item $E_{\valids} = \infty$: In this case we must have $|\valids| = \infty$. We proceed as in case (\ref{case:pump-v}) of Theorem \ref{theorem-dfa-scheme}, but using Lemma \ref{lemma:symb-pump-samp} with $N = \ceil{1/\rho}$ in place of Lemma \ref{lemma-dfa-pump-samp} to obtain a $(0, \rho)$-improviser.

\item $1/\rho \le E_{\valids} < \infty$: Then $1/ \rho(1+\tau) \le E_{\valids} / (1+\tau) \le |\valids|$. We proceed as in case (\ref{case:unif-v}) of Theorem \ref{theorem-dfa-scheme}, using Lemma \ref{lemma:symb-unif-samp} in place of Lemma \ref{lemma-dfa-unif-samp}. Since we are using an almost-uniform generator instead of a uniform one, some words could have probability as high as $(1+\tau) / |\valids| \le (1+\tau)^2 \rho$, and so this gives us a $(0, (1+\tau)^2 \rho)$-improviser.

\item $(1-\epsilon)/\rho \le E_{\valids} < 1/\rho$ and $E_{\improvs} = \infty$: Then $(1-\epsilon)/\rho(1+\tau) \le E_{\valids} / (1+\tau) \le |\valids|$, and $|\improvs| = \infty$. We proceed along the same lines as case (\ref{case:pump-i-unif-v}) of Theorem \ref{theorem-dfa-scheme}. We use Lemma \ref{lemma:symb-unif-samp} as in the previous case to generate a probabilistic algorithm $S$ almost-uniformly sampling from $L(\mathcal{\valids})$. Defining $M = \ceil{(1/\rho - E_{\valids})/(1+\tau)}$, we have $\infty = |L(\mathcal{B})| > M \ge 1$. We can use Lemma \ref{lemma:symb-pump-samp} in place of Lemma \ref{lemma-dfa-pump-samp} to get a probabilistic algorithm $S'$ uniformly sampling from a subset of $L(\mathcal{B})$ of size $M$. Let $G$ be a probabilistic algorithm which with probability $\rho E_{\valids}$ executes $S$, and otherwise executes $S'$. Then since $L(\mathcal{\valids}) = \valids$ and $L(\mathcal{B}) = \improvs \setminus \valids$ are disjoint, every word generated by $G$ has probability either at most $(\rho E_{\valids}) \cdot (1+\tau) / |\valids| \le (1+\tau)^2 \rho$ (if it is in $\valids$) or at most $(1 - \rho E_{\valids}) / M = (1 - \rho E_{\valids}) / \ceil{(1/\rho - E_{\valids})/(1+\tau)} = (1 - \rho E_{\valids}) / \ceil{(1 - \rho E_{\valids}) / (1+\tau)\rho} \le (1+\tau)\rho$ (if it is in $\improvs \setminus \valids$). Also $G$ outputs a member of $\valids$ with probability $\rho E_{\valids} \ge 1 - \epsilon$, so $G$ is an $(\epsilon, (1+\tau)^2 \rho)$-improviser and we return it.

\item $(1-\epsilon)/\rho \le E_{\valids} < 1/\rho \le E_{\improvs} < \infty$: Then $(1-\epsilon)/\rho(1+\tau) \le E_{\valids} / (1+\tau) \le |\valids|$, and $1/\rho(1+\tau) \le E_{\improvs} / (1+\tau) \le |\improvs|$. As in the previous case, use Lemma \ref{lemma:symb-unif-samp} to produce a probabilistic algorithm $S$ almost-uniformly sampling from $L(\mathcal{\valids})$. Since $L(\mathcal{\improvs})$ is finite, we can use the same technique to get a probabilistic algorithm $S'$ almost-uniformly sampling from $L(\mathcal{\improvs})$. Let $G$ be a probabilistic algorithm which with probability $1-\epsilon$ executes $S$, and otherwise executes $S'$. Then $G$ generates each $w \in \valids$ with probability at most $[ (1-\epsilon) \cdot (1+\tau) / |\valids| ] + [ \epsilon \cdot (1+\tau) / |\improvs| ] \le (1+\tau)^2 \rho + (1+\tau)^2 \epsilon \rho = (1+\tau)^2 (1+\epsilon) \rho$, and each $w \in \improvs \setminus \valids$ with probability at most $\epsilon \cdot (1+\tau) / |\improvs| \le (1+\tau)^2 \epsilon \rho$. Furthermore $G$ outputs a member of $\valids$ with probability at least $1-\epsilon$, so $G$ is an $(\epsilon, (1+\tau)^2 (1+\epsilon) \rho)$-improviser and we return it.

\item $E_{\improvs} < 1/\rho$ or $E_{\valids} < (1-\epsilon) / \rho$: It is possible that the problem is not $(\epsilon, \rho/(1+\tau))$-feasible, so the procedure returns $\bot$.
\end{enumerate}
In the cases where an almost-uniform generator is used, there is some constant probability that the generator will fail. If that happens, the improviser just runs the generator again: since the failure probability is a fixed constant, so is the expected number of repetitions needed, and thus the expected runtime of the improviser is just multiplied by an overall constant.

Now if $\mathcal{C}$ is $(\epsilon, \rho/(1+\tau))$-feasible, by Theorem \ref{theorem:feasibility} we have $1/\rho \le |\improvs| / (1+\tau) \le E_{\improvs}$ and $(1-\epsilon) / \rho \le |\valids| / (1+\tau) \le E_{\valids}$ with probability at least $1-\delta$. So with probability $1-\delta$ case (\ref{case:symb-infeasible}) does not happen, and the procedure returns an $(\epsilon, (1+\tau) \rho)$-improviser.
\end{proof}

\end{document}